\documentclass[reprint,amsmath,amssymb,aps,prx,floatfix,longbibliography]{revtex4-2}

\usepackage{graphicx}
\usepackage{dcolumn}
\usepackage{bm}
\usepackage{hyperref}
\usepackage{xcolor}
\usepackage{booktabs}
\usepackage{amsthm}
\hypersetup{colorlinks=true,linkcolor=blue,citecolor=blue,urlcolor=blue}

\newtheorem{theorem}{Theorem}
\newtheorem{lemma}{Lemma}

\newcommand{\ket}[1]{\left|#1\right\rangle}
\newcommand{\bra}[1]{\left\langle#1\right|}
\newcommand{\mathds}[1]{\mathbb{#1}}

\newcommand{\HSh}{H_{\mathrm{Sh}}^{\mathrm{ext}}}
\newcommand{\Hcert}{H_{\mathrm{cert}}}
\newcommand{\Rcert}{\mathcal{R}_{\mathrm{cert}}}
\newcommand{\meff}{\mu_{\mathrm{eff}}}
\newcommand{\dKL}{D_{\mathrm{KL}}}

\begin{document}

\title{Squeezed-state semi-device-independent quantum randomness generation}

\author{Hamid Tebyanian}
\affiliation{School of Physical and Chemical Sciences, Queen Mary University of London, Mile End Road, London E1 4NS, United Kingdom}

\date{\today}

\begin{abstract}
This paper investigates semi-device-independent quantum randomness generation with a trusted binary pure-state source and an untrusted binary detector whose side information is classical. We derive a closed-form Shannon-rate expression for this setting, depending only on the trusted Gram overlap of the two source states and the observed symmetric error probability. The key point is that the full binary-qubit POVM optimisation must include the two deterministic extreme points omitted by the projective-only treatment; including them gives a substantially lower, and correct, certified rate.
The closed form is an unconditional upper bound on the certified asymptotic i.i.d.\ Shannon rate, and becomes tight on a numerically verified dual-feasibility region containing all operating points used in the paper. Outside this region the same expression remains an upper bound. We then apply the result to squeezed-coherent BPSK sources, showing how squeezing changes the trade-off between state distinguishability and certified randomness in the lossless and lossy regimes. Finally, we clarify the adversary model if the adversary is allowed to hold a detector-purification register that tags the outcome.
\end{abstract}

\maketitle

\section{Introduction}

A practical quantum random number generator (QRNG) must combine a generation rate suitable for cryptographic use with a security argument that remains valid when part of the device is untrusted. Fully device-independent QRNGs~\cite{Pironio2010, Liu2018, Bierhorst2018} achieve this in principle, but they need loophole-free Bell violations and so stay limited to kHz-class rates. By contrast, fully trusted hardware QRNGs~\cite{HerreroCollantes2017, Mannalath2023, Moradpour2025} can reach 100-Gbit/s-class speeds~\cite{Bruynsteen2023}, but rely on hardware certification by the manufacturer.

Semi-device-independent (semi-DI) QRNG protocols occupy this intermediate regime by replacing full device trust with explicit physical assumptions, such as a dimension bound, an energy bound, a trusted overlap/indistinguishability constraint, or a trusted measurement model~\cite{Brask2017, VanHimbeeck2017, Rusca2020, Drahi2020, Foletto2021, Tebyanian2021, Tebyanian2024, Bhavsar2026, Genzini2026, Tebyanian2024TimeBin}. A related but distinct source-device-independent line instead treats the source as untrusted while relying on a characterised measurement stage, for example in homodyne or heterodyne continuous-variable QRNGs~\cite{Marangon2017, Avesani2018, Smith2019, Avesani2021, Cizauskas2026}. The present work belongs to the prepare-and-measure semi-DI class: the source is trusted only through the Gram overlap of a binary pure-state pair, while the binary detector is untrusted and may depend on a classical hidden variable $\lambda$ held by the adversary. Under this specific model, the asymptotic Shannon randomness of the user's output is a linear program (LP) over binary POVMs acting on the source two-dimensional subspace, with constraints fixed by the observed marginals.

Although this LP has been solved in closed form for the strictly projective subset of binary qubit POVMs~\cite{VanHimbeeck2017, Tebyanian2021}, the full binary-qubit POVM convex set also contains two rank-deficient extreme points, namely the deterministic POVMs $M_0 = 0$ and $M_0 = I$; these deterministic components contribute zero conditional entropy and shift the feasible set once the marginal-matching constraints are imposed, so restricting the optimisation to projective components overestimates the certified rate. The overestimate is large in the rate-positive interior of the $(\delta, q)$ plane: at $\delta = 0.5$, $q = 0.15$ the projective-only formula returns $0.25$ bits whereas the correct certified rate is $0.06$ bits.

Here we solve the extended binary-qubit POVM optimisation for this prepare-and-measure semi-DI model, including the deterministic POVMs omitted by the projective-only treatment. This gives a closed-form Shannon-rate expression that is an unconditional upper bound in the full feasible region and is tight on a numerically verified dual-feasibility region containing all operating points used below. We then apply the result to squeezed-coherent BPSK sources, where squeezing improves state distinguishability but can reduce certified randomness, leading to a direct trade-off in both the lossless and lossy regimes.

Compared with earlier semi-DI QRNG protocols based on unambiguous state discrimination, energy or dimension bounds, and numerical source-indistinguishability analyses~\cite{Brask2017, VanHimbeeck2017, Rusca2020, Drahi2020, Foletto2021, Tebyanian2021, Tebyanian2024, Bhavsar2026, Genzini2026}, the contribution here is a closed-form Shannon-rate analysis for a squeezed-coherent BPSK source under a trusted Gram-overlap assumption and an untrusted binary detector. This is, to our knowledge, the first semi-DI squeezed-state QRNG rate formula in this trust model. The improvement over the projective-only treatment is that the full binary-qubit POVM optimisation is solved, including the deterministic POVMs $M_0=0$ and $M_0=I$; these deterministic POVMs change the adversarial decomposition and can substantially lower the certified rate, as shown by the factor-four correction in Fig.~\ref{fig:validation}. This separates the present result from source-device-independent homodyne and heterodyne QRNGs, which use the opposite trust split and certify randomness through a characterised measurement stage~\cite{Marangon2017, Avesani2018, Smith2019, Avesani2021, Cizauskas2026}.


\section{protocol}\label{sec:setup}

The user holds a binary input $X \in \{0, 1\}$ with prior $P(X = 0) = P(X = 1) = 1/2$ and prepares, for each value of $X$, a pure state $\ket{\psi_X}$ on a single optical mode; the two states span a two-dimensional subspace, which we identify with $\mathbb{C}^2$ for the rest of the analysis, and their Gram overlap $\delta := \langle\psi_0 | \psi_1\rangle$ is taken real and positive after a global phase choice. This overlap is the only source parameter entering the rate formula, while the binary detector is untrusted and is modelled as a classical mixture of POVMs $\{M_{0|\lambda}, M_{1|\lambda}\}$ indexed by a hidden variable $\lambda$ with distribution $p_\lambda$, which is held by the adversary.

Let $q := P(B = 1 \mid X = 0)$; under symmetric BPSK $P(B = 0 \mid X = 1) = q$ also holds, so $q$ is the single-symbol error probability on either input. In an implementation the two marginals must be estimated separately; the single-parameter model applies after verifying this symmetry, otherwise the corresponding two-error LP has to be used. The LP marginal constraints are
\begin{equation}
  \sum_\lambda p_\lambda\,\langle\psi_0 | M_{0 | \lambda} | \psi_0\rangle = 1 - q,
  \quad
  \sum_\lambda p_\lambda\,\langle\psi_1 | M_{0 | \lambda} | \psi_1\rangle = q.
  \label{eq:marginals}
\end{equation}

The asymptotic i.i.d.\ Shannon randomness of $B$ against the classical-$\lambda$ adversary is
\begin{equation}
  \Hcert(\delta, q) :=
  \min_{\{p_\lambda, M_{b|\lambda}\}}\
  \sum_\lambda p_\lambda\,h\!\bigl(\langle\psi_0 | M_{0 | \lambda} | \psi_0\rangle\bigr),
  \label{eq:Hcert-def}
\end{equation}
subject to Eq.~\eqref{eq:marginals}, with $h(p) = -p \log_2 p - (1-p)\log_2(1-p)$ the binary entropy in bits. Appendix~\ref{app:threat} writes Eq.~\eqref{eq:Hcert-def} as a classical conditional entropy under the classical-$\lambda$ adversary, and shows that the detector-purification adversary that holds an orthogonal outcome-tagging environment register reduces the certified rate to zero. Operationally, Eq.~\eqref{eq:Hcert-def} certifies the entropy of the generation input $X=0$, while the $X=1$ marginal is used as a constraint on Eve's feasible decompositions. For a protocol that extracts randomness from both inputs, the relevant quantity would instead be the averaged optimisation $\Hcert^{\rm av}(\delta,q)=H(B|X,\Lambda)$, namely
\begin{equation}
\begin{aligned}
  & \Hcert^{\rm av}(\delta,q) \\:= &\min_{\{p_\lambda,M_{b|\lambda}\}}
  \sum_\lambda p_\lambda
  \frac{
  h\!\bigl(\langle\psi_0|M_{0|\lambda}|\psi_0\rangle\bigr)
  +
  h\!\bigl(\langle\psi_1|M_{0|\lambda}|\psi_1\rangle\bigr)
  }{2},
 \end{aligned}
\end{equation}
subject to the same marginal constraints. This is a different optimisation from Eq.~\eqref{eq:Hcert-def} and is not evaluated here. The closed form derived below is therefore a randomness guarantee against the classical-$\lambda$ adversary only; the classical nature of the side information is an explicit semi-DI trust assumption, not a property certified by the observed marginals.

The closed form depends on the source only through $\delta$, so any trusted pure-state pair with this overlap admits the same rate formula; the squeezed-coherent BPSK source is used here as the concrete optical realisation that fixes the operating points.

The two states are $\ket{\psi_x} = D\!\bigl((-1)^x \sqrt N\bigr)\,S(r)\,\ket 0$ for $x \in \{0, 1\}$, with $D(\alpha)$ the displacement operator, $S(r) = \exp\!\left[\tfrac{r}{2}(a^2 - a^{\dagger 2})\right]$ the single-mode squeezing operator, $N \ge 0$ the squared displacement amplitude, $r \ge 0$ the squeezing parameter, and the squeezing direction aligned with the displacement axis~\cite{Yuen1976,Caves1981,Weedbrook2012}. Using $S^\dagger(r) D(\alpha) S(r) = D(\alpha e^{r})$,
\begin{equation}
  \delta = \langle\psi_0|\psi_1\rangle = e^{-2 N e^{2r}},
  \label{eq:overlap}
\end{equation}
and the mean photon number of each branch is $\bar n = N + \sinh^2 r$.
For a homodyne-style binary detector of efficiency $\eta$ that bins the squeezed-quadrature outcome at zero, the variance of the quadrature is $\sigma^2 = \tfrac{1}{2}(\eta e^{-2r} + 1 - \eta)$ and the honest error probability is
\begin{equation}
  q_{\rm hon}(N, r, \eta) = \tfrac{1}{2}\bigl(1 - \mathrm{erf}\sqrt{2\,\meff^{\rm loss}}\bigr),\;
  \meff^{\rm loss} := \frac{\eta N}{\eta e^{-2r} + 1 - \eta}.
  \label{eq:q-honest-lossy}
\end{equation}
At $\eta = 1$ this reduces to $q_{\rm hon}(N, r) = \tfrac{1}{2}(1 - \mathrm{erf}\sqrt{2 N e^{2r}})$. Equations~\eqref{eq:overlap} and~\eqref{eq:q-honest-lossy} assume that the displacement and squeezing axes are phase-aligned. If the squeezing axis is misaligned from the displacement axis by an angle $\phi$, the overlap exponent is modified to
\begin{equation}
  \delta(\phi)
  =
  \exp\!\left[-2N\left(e^{2r}\cos^2\phi+e^{-2r}\sin^2\phi\right)\right],
\end{equation}
so phase mismatch weakens the distinguishability gain from squeezing and must be included in the trusted value of $\delta$. The closed-form theorem itself is unchanged once the correct trusted $\delta$ and observed $q$ are used. Equation~\eqref{eq:q-honest-lossy} is a prediction for the honest device, not a security claim about a device of internal efficiency $\eta$; the bound of Theorem~\ref{thm:main} applies to any observed $q$ in the strong-feasibility interval and is the certified rate inside $\Rcert$, while Eq.~\eqref{eq:q-honest-lossy} is used only to locate honest operating points on the rate surface.

\section{Closed-form certified rate}\label{sec:closed-form}

The optimisation of Eq.~\eqref{eq:Hcert-def} is a semi-infinite LP over the convex set of classical-$\lambda$ binary POVMs. On the source two-dimensional subspace, the binary qubit POVM convex set has three kinds of extreme points: the rank-one projectives $\Pi_\theta$ parametrised by $\theta \in [0, \pi)$, the deterministic POVM $M_0 = 0$ with moment $(\langle\psi_0|M_0|\psi_0\rangle, \langle\psi_1|M_0|\psi_1\rangle) = (0, 0)$, and the deterministic POVM $M_0 = I$ with moment $(1, 1)$. The projective component at angle $\theta$ contributes moment $(\cos^2\theta, \cos^2(\theta - \theta_*))$ and entropy $h(\cos^2\theta)$, with $\theta_* := \arccos\delta$. Both deterministic endpoints have zero entropy~\cite{Busch2016}. These are the only extreme points: after diagonalising a binary qubit effect $0\le M_0\le I$, any eigenvalue strictly inside $(0,1)$ can be perturbed in both directions while staying inside the effect interval, so extremality requires each eigenvalue to be either $0$ or $1$. The possible extreme effects are therefore $0$, $I$, and the rank-one projections. Because $h$ is concave, any non-extreme binary POVM element can be written as a convex combination of these components, and splitting $\lambda$ accordingly does not raise the objective, since $h\bigl(\sum_k c_k x_k\bigr) \ge \sum_k c_k h(x_k)$~\cite{CoverThomas2006}. Thus the minimisation in Eq.~\eqref{eq:Hcert-def} loses nothing by restricting each $M_{b|\lambda}$ to an extreme point, which is why the problem reduces to an LP over distributions on the projective and deterministic endpoints.

The strong-feasibility argument of Appendix~\ref{app:strong-feas} shows that the LP is feasible for $q$ in the symmetric interval $[q_b, 1 - q_b]$ with
\begin{equation}
  q_b := \tfrac{1}{2}\bigl(1 - \sqrt{1 - \delta^2}\bigr).
  \label{eq:qb}
\end{equation}
Inside this interval, the rate-positive part is $(q_b, q_z) \cup (1 - q_z, 1 - q_b)$ with
\begin{equation}
  q_z := \frac{\delta^2}{1 + \delta^2},
  \label{eq:qz}
\end{equation}
and the rate-zero plateau is $q \in [q_z, 1 - q_z]$, which is non-empty because $q_z<1/2$ for every $\delta\in(0,1)$. The two rate-positive segments are equivalent under $q \leftrightarrow 1 - q$ (the BPSK relabelling $X \leftrightarrow 1 - X$); we work with $q \in (q_b, q_z)$ throughout.

\begin{figure}
  \centering
  \includegraphics[width=\columnwidth]{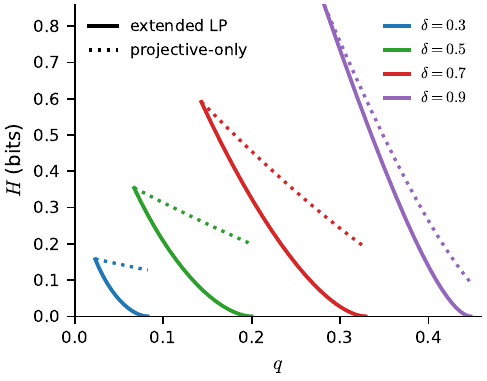}
  \caption{Extended-LP closed form $\HSh = r\cdot h(p)$ (solid) and projective-only LP value $h(p_{\rm proj})$ (dotted), on the rate-positive interval $q \in (q_b, q_z)$ at $\delta \in \{0.3, 0.5, 0.7, 0.9\}$.}
  \label{fig:validation}
\end{figure}
\begin{theorem}[Certified Shannon rate, extended LP]\label{thm:main}
Let $\delta \in (0, 1)$ and $q \in [q_b, 1 - q_b]$ with $q_b, q_z$ from Eqs.~\eqref{eq:qb}, \eqref{eq:qz}. Define
\begin{align}
  r(\delta, q) &:= \frac{1 - 2\delta\sqrt{q(1-q)}}{1 - \delta^2}, \label{eq:r-def}\\
  p(\delta, q) &:= \frac{(1-q)(1-\delta^2)}{1 - 2\delta\sqrt{q(1-q)}}. \label{eq:p-def}
\end{align}
\begin{equation}
  \HSh(\delta, q) := \begin{cases}
    r(\delta, q)\,h\bigl(p(\delta, q)\bigr) & q \in [q_b, q_z),\\
    0 & q \in [q_z, 1 - q_z],\\
    r(\delta, 1-q)\,h\bigl(p(\delta, 1-q)\bigr) & q \in (1 - q_z, 1 - q_b].
  \end{cases}
  \label{eq:HSh-def}
\end{equation}
Unconditionally, $\Hcert(\delta,q)\le \HSh(\delta,q)$. This upper bound is realised as a feasible primal value by the explicit Eve strategy of Lemma~\ref{lem:primal} on the rate-positive interior, and by a three-component zero-entropy decomposition on the rate-zero plateau. Equality holds on the rate-zero plateau and on the dual-feasibility region
\begin{equation}
\Rcert := \bigl\{(\delta, q):\ L_{\delta, q}(\theta) \le h(\cos^2\theta)\ \forall\,\theta \in [0, \pi)\bigr\},
\end{equation}
where $L_{\delta, q}(\theta)=\alpha\cos^2\theta+\beta\cos^2(\theta-\theta_*)+\gamma$ is the candidate dual witness of Lemma~\ref{lem:dual}. Thus $\Hcert(\delta,q)=\HSh(\delta,q)$ for every $(\delta,q)\in\Rcert$. The global dual inequality defining $\Rcert$ is verified numerically in Appendix~\ref{app:certified-region}, with the analytic local condition stated there as a necessary first-order requirement. Every tabulated and optimised operating point used for certified-rate claims in this paper lies in $\Rcert$ at the numerical precision reported.
\end{theorem}

The closed-form rate is invariant under $q \leftrightarrow 1 - q$, because $h(p) = h(1 - p)$ and the LP is invariant under Eve's outcome relabelling $B \leftrightarrow 1 - B$; we therefore fix the convention $q \in (q_b, q_z) \subset (0, 1/2)$, in which case $p \in (1 - q_b, 1)$ takes values close to $1$ on the rate-positive interior, with $p \to 1$ as $q \to q_z^-$ and $p \to 1 - q_b$ as $q \to q_b^+$.

Figure~\ref{fig:validation} compares the corrected closed form with the projective-only value $h(p_{\rm proj})$ across $\delta \in \{0.3, 0.5, 0.7, 0.9\}$, showing for example that at $\delta = 0.5$ and $q = 0.15$, close to the plateau edge $q_z = 0.2$, the projective-only value is about four times the corrected one. Figure~\ref{fig:contour} maps $\HSh$ on the rate-positive interior of the $(\delta, q)$ plane.

\begin{figure}
  \centering
  \includegraphics[width=\columnwidth]{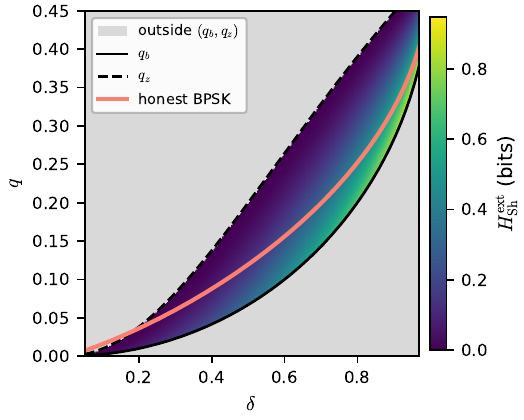}
  \caption{$\HSh(\delta, q)$ on the rate-positive interior $q \in (q_b, q_z)$. Curves: $q_b$ (solid black), $q_z$ (dashed black), honest BPSK $q = (1 - \mathrm{erf}\sqrt{2\meff})/2$ at $\delta = e^{-2\meff}$ (salmon). Gray: outside the displayed positive-rate branch; below $q_b$ the observation is outside the feasible moment region, while above $q_z$ the deterministic/projective zero-cost decomposition gives the rate-zero plateau.}
  \label{fig:contour}
\end{figure}

\subsection{Primal construction}

\begin{lemma}[Primal Eve strategy]\label{lem:primal}
For $(\delta, q)$ with $q \in (q_b, q_z)$, let $r$ and $p$ be as in Eqs.~\eqref{eq:r-def}, \eqref{eq:p-def}, and set
\begin{equation}
  \theta_0 := \pi - \arccos\sqrt{p} \in (\pi/2, \pi),\qquad w_0 := 1 - r.
  \label{eq:theta-0}
\end{equation}
Then $0 \le p \le 1$, $0 \le r \le 1$, and the binary POVM
\begin{equation}
  M_b = r\,\Pi_{\theta_0}^{(b)} + w_0\,\mathds{1}_{(0,0)}^{(b)},
  \qquad b \in \{0, 1\},
\end{equation}
with $\Pi_{\theta_0}^{(0)}$ the rank-one projection onto the state at angle $\theta_0$, $\Pi_{\theta_0}^{(1)} = I - \Pi_{\theta_0}^{(0)}$, $\mathds{1}_{(0,0)}^{(0)} = 0$, $\mathds{1}_{(0,0)}^{(1)} = I$, satisfies Eq.~\eqref{eq:marginals} term by term and has single-symbol conditional entropy $r\cdot h(p) = \HSh(\delta, q)$. Hence $\Hcert(\delta, q) \le \HSh(\delta, q)$.

For $q \in [q_z, 1/2]$ on the rate-zero plateau, the three-component decomposition
\begin{equation}
  w_p = \frac{1 - 2q}{1 - \delta^2},\
  w_1 = \frac{q(1+\delta^2) - \delta^2}{1 - \delta^2},\
  w_0 = q,
  \label{eq:plateau-weights}
\end{equation}
over the projective at $\theta = 0$ (moment $(1, \delta^2)$) and the deterministic endpoints $(1, 1)$ and $(0, 0)$ has $w_p, w_1, w_0 \ge 0$, $w_p + w_1 + w_0 = 1$, reproduces the marginals $(1 - q, q)$, and has zero entropy on every component. For $q \in [1/2, 1 - q_z]$ the mirrored decomposition
\begin{equation}
  w_p' = \frac{2q - 1}{1 - \delta^2},\
  w_1' = 1 - q,\
  w_0' = \frac{1 - q(1+\delta^2)}{1 - \delta^2},
  \label{eq:plateau-weights-mirror}
\end{equation}
over the projective at $\theta = \pi/2$ (moment $(0, 1 - \delta^2)$) and the same two deterministic endpoints also has non-negative weights summing to one, reproduces $(1 - q, q)$, and has zero entropy. Since conditional Shannon entropy is non-negative, this gives $\Hcert(\delta,q)=0=\HSh(\delta,q)$ on the entire plateau.
\end{lemma}

\begin{proof}
Non-negativity of $p$ requires $1 - 2\delta\sqrt{q(1-q)} > 0$, automatic from $\sqrt{q(1-q)} \le 1/2 < 1/(2\delta)$ for any $\delta < 1$. The upper bound $p\le1$ follows from 
\[ 
1-p=\frac{\bigl(\sqrt q-\delta\sqrt{1-q}\bigr)^2}{1-2\delta\sqrt{q(1-q)}}\ge0 . 
\] 
Non-negativity of $w_0 = 1 - r$ requires $r \le 1$, equivalent to $\delta \le 2\sqrt{q(1-q)}$; squaring gives $4q(1-q)\ge\delta^2$, or $|1-2q|\le\sqrt{1-\delta^2}$, exactly the BPSK chord condition derived in Appendix~\ref{app:strong-feas}. The branch choice $\theta_0 = \pi - \arccos\sqrt{p}$ in Eq.~\eqref{eq:theta-0} corresponds to $\cos\theta_0 = -\sqrt{p}$ and $\sin\theta_0 = \sqrt{1 - p}$.

For the marginals: $\cos^2\theta_0 = p$, so $r\cdot p = (1-q)$ by direct algebra on Eqs.~\eqref{eq:r-def}, \eqref{eq:p-def}. For the second marginal,
\begin{equation}
\begin{aligned}
  &\cos(\theta_0 - \theta_*) =\cos\theta_0\cos\theta_* + \sin\theta_0\sin\theta_*
  \\ &= -\delta\sqrt p + \sqrt{(1-p)(1-\delta^2)},
\end{aligned}
\end{equation}
which gives
\begin{equation}
  \cos^2(\theta_0 - \theta_*) = \delta^2 p - 2\delta\sqrt{p(1-p)(1-\delta^2)} + (1 - p)(1 - \delta^2).
\end{equation}
Substituting Eq.~\eqref{eq:p-def} into the right-hand side gives, after a short calculation, $r\cdot\cos^2(\theta_0 - \theta_*) = q$. Hence the strategy realises Eq.~\eqref{eq:marginals}. The branch $\theta_0 \in (0, \pi/2)$, by contrast, gives $\cos(\theta_0 - \theta_*) = +\delta\sqrt p + \sqrt{(1-p)(1-\delta^2)}$ and $r\cdot\cos^2(\theta_0 - \theta_*) \ne q$ for $q < 1/2$; the correct active branch is $(\pi/2, \pi)$.

The single-symbol conditional entropy of the strategy is $r\cdot h(\cos^2\theta_0) + w_0\cdot h(0) = r\cdot h(p)$, since $h(0) = 0$ on the deterministic endpoint $(0, 0)$.

For the rate-zero plateau, the weights of Eq.~\eqref{eq:plateau-weights} are non-negative on $q \in [q_z, 1/2]$: $w_p \ge 0$ since $q \le 1/2$; $w_1 \ge 0$ iff $q \ge \delta^2/(1+\delta^2) = q_z$; $w_0 = q \ge 0$ always. They sum to one by direct algebra. The components' moments $(1, \delta^2)$, $(1, 1)$, $(0, 0)$ combine to give marginals
\begin{align}
  u &= w_p\cdot 1 + w_1\cdot 1 + w_0\cdot 0 = 1 - q,\\
  v &= w_p\cdot\delta^2 + w_1\cdot 1 + w_0\cdot 0 = q.
\end{align}
Each component has zero entropy: the projective at $\theta = 0$ has $h(\cos^2 0) = h(1) = 0$, and the deterministic endpoints have entropy zero by definition. For the mirror $q \in [1/2, 1 - q_z]$, the weights of Eq.~\eqref{eq:plateau-weights-mirror} satisfy the same non-negativity bounds: $w_p' \ge 0$ since $q \ge 1/2$; $w_1' \ge 0$ since $q \le 1$; $w_0' \ge 0$ iff $q \le 1/(1+\delta^2) = 1 - q_z$. The components' moments $(0, 1-\delta^2)$, $(1, 1)$, $(0, 0)$ combine to give $(u, v) = (w_1', \, w_p'(1-\delta^2) + w_1') = (1-q, q)$, again with zero entropy on every component.
\end{proof}

\subsection{Dual certificate}

\begin{lemma}[Dual witness]\label{lem:dual}
For $(\delta, q)$ with $q \in (q_b, q_z)$ and $\theta_0$ as in Eq.~\eqref{eq:theta-0}, consider the affine function
\begin{equation}
  L_{\delta, q}(\theta)
  :=
  \alpha\cos^2\theta
  + \beta\cos^2(\theta - \theta_*)
  + \gamma .
\end{equation}
The coefficients are fixed by imposing value and slope tangency at the active projective point, and by setting the affine offset so that the zero-entropy endpoint $(u,v)=(0,0)$ is saturated:
\begin{align}
  \alpha\cos^2\theta_0 + \beta\cos^2(\theta_0 - \theta_*) &= h(p),\label{eq:dual-val}\\
  \alpha\sin(2\theta_0) + \beta\sin(2(\theta_0 - \theta_*)) &= h'(p)\sin(2\theta_0),\label{eq:dual-slope}\\
  \gamma &= 0.\label{eq:dual-gamma}
\end{align}
This linear system has a unique solution on the rate-positive interior. Its determinant is
\begin{equation}
\begin{split}
  D &:= \cos^2\theta_0\,\sin\bigl(2(\theta_0 - \theta_*)\bigr)
       - \cos^2(\theta_0 - \theta_*)\,\sin(2\theta_0) \\
    &= -2\sin\theta_*\,\cos\theta_0\,\cos(\theta_0 - \theta_*),
\end{split}
  \label{eq:det-formula}
\end{equation}
which is non-zero for $q \in (q_b,q_z)$. With this choice, $L_{\delta,q}$ is tight at the active projective, has the same slope as $h(\cos^2\theta)$ at that point, and is tight at the deterministic endpoint $(0,0)$, namely
\begin{equation}
\begin{aligned}
  &L_{\delta,q}(\theta_0)=h(p),\\
  &L'_{\delta,q}(\theta_0)=-h'(p)\sin(2\theta_0),\\
  &L_{\delta,q}|_{(u,v)=(0,0)}=0 .
  \end{aligned}
\end{equation}
At the other zero-entropy endpoint $(u,v)=(1,1)$ it has non-positive slack,
\begin{equation}
  L_{\delta,q}|_{(u,v)=(1,1)}
  = \alpha+\beta+\gamma
  = \alpha+\beta
  \le 0 = h(1),
\end{equation}
with equality only in the limit $q\to q_z^-$, as proved in Appendix~\ref{app:slack-sign}. Moreover, evaluated at the observed marginals, the same affine functional gives
\begin{equation}
  \alpha(1-q)+\beta q+\gamma
  =
  r(\delta,q)\,h(p)
  =
  \HSh(\delta,q).
\end{equation}
Therefore, whenever the global pointwise inequality
\begin{equation}
  L_{\delta,q}(\theta)\le h(\cos^2\theta)
  \qquad \forall\,\theta\in[0,\pi)
\end{equation}
holds, the dual value matches the primal value of Lemma~\ref{lem:primal}. Hence weak duality gives
$\Hcert(\delta,q)=\HSh(\delta,q)$ on $\Rcert$.
\end{lemma}

\begin{proof}
The first two equations in the defining system impose value and slope tangency at the active projective atom, while $\gamma=0$ imposes tightness at the deterministic atom $(0,0)$. The determinant formula in Eq.~\eqref{eq:det-formula} follows from elementary trigonometric identities applied to $\cos^2\theta_0\,\sin\bigl(2(\theta_0-\theta_*)\bigr)-\cos^2(\theta_0-\theta_*)\,\sin(2\theta_0)$, using $\sin\bigl(2(\theta_0-\theta_*)\bigr)-\sin(2\theta_0)=-2\sin\theta_*\,\cos(2\theta_0-\theta_*)$ and $\cos^2(\theta_0-\theta_*)-\cos^2\theta_0=\sin\theta_*\,\sin(2\theta_0-\theta_*)$, then collecting via sum-to-product. The determinant is strictly negative on the rate-positive interior because $\sin\theta_*>0$, $\cos\theta_0<0$ by the branch choice, and $\cos(\theta_0-\theta_*)<0$. The last sign is equivalent to $p>1-\delta^2$, which holds on the interior since $p>1-q_b>1-\delta^2$. Hence the dual coefficients are well-defined.

The slack condition at the deterministic atom $(1,1)$, namely $\alpha+\beta\le0$, is proved in Appendix~\ref{app:slack-sign}. It remains only to check the dual value at the observed marginals. From Lemma~\ref{lem:primal}, $r\cos^2\theta_0=1-q$ and $r\cos^2(\theta_0-\theta_*)=q$. Therefore
\begin{align}
  \alpha(1-q)+\beta q+\gamma
   &= r\bigl[\alpha\cos^2\theta_0+\beta\cos^2(\theta_0-\theta_*)\bigr]+0\notag\\
   &= r\,h(p)
    = \HSh(\delta,q),
\end{align}
where the second equality uses the value-tangency condition. Thus, whenever the global pointwise constraint $L_{\delta,q}(\theta)\le h(\cos^2\theta)$ holds for all $\theta\in[0,\pi)$, the affine functional is a feasible dual witness with value $\HSh(\delta,q)$. Combined with the primal construction of Lemma~\ref{lem:primal}, weak duality gives equality on $\Rcert$. Appendix~\ref{app:certified-region} states the necessary local condition and verifies the global inequality numerically on the operating region used in this paper.
\end{proof}

Figure~\ref{fig:dual-cert} shows the dual witness at $\delta = 0.6$ and the residual $g(\theta) := h(\cos^2\theta) - L_{\delta, q}(\theta)$ at several $\delta$ along honest BPSK.

\begin{figure}
  \centering
  \includegraphics[width=\linewidth]{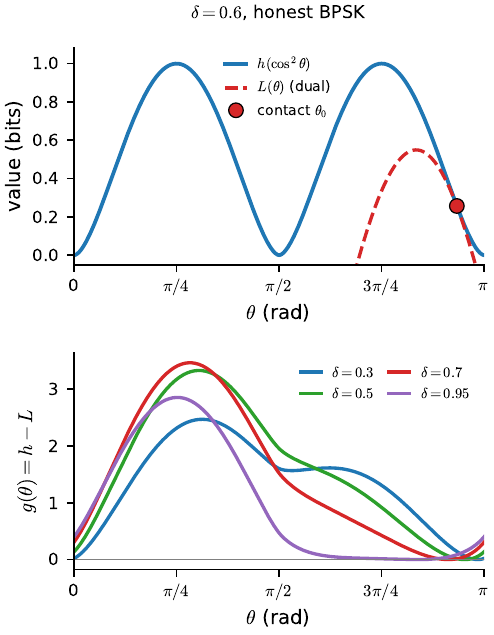}
  \caption{Top: $h(\cos^2\theta)$ (solid) and dual witness $L_{\delta, q}(\theta)$ (dashed) over the full angular range $\theta\in[0,\pi]$ at $\delta = 0.6$ for the honest BPSK observation, using the active branch $\theta_0=\pi-\arccos\sqrt p\in(\pi/2,\pi)$. Filled marker: tangent contact at $\theta_0$. Bottom: residual $g(\theta) = h(\cos^2\theta) - L_{\delta, q}(\theta)$ at $\delta \in \{0.3, 0.5, 0.7, 0.95\}$ for the honest (nominal) BPSK observation.}
  \label{fig:dual-cert}
\end{figure}

\subsection{Optimal adversarial decomposition}

The optimal Eve strategy of Lemma~\ref{lem:primal} mixes a single projective at $\theta_0 \in (\pi/2, \pi)$ with weight $r$ and the deterministic endpoint $(0, 0)$ with weight $1-r$, so in the moment plane $(u, v) = (\langle\psi_0|M_0|\psi_0\rangle, \langle\psi_1|M_0|\psi_1\rangle)$ the observation $(1-q,q)$ lies on the chord from $(\cos^2\theta_0,\cos^2(\theta_0-\theta_*))$ to $(0,0)$ with chord parameter $r$. This one-projective-plus-deterministic decomposition is exactly what the projective-only formula misses, since the projective-only LP instead uses the best two-projective decomposition of the same observation and therefore gives a strictly larger conditional entropy.

\section{Squeezing trade-off}\label{sec:opt-squeeze}

For fixed $\bar n = N + \sinh^2 r$ at lossless detection, the closed form depends on $\bar n$ and $r$ only through $\meff(r) := N(r)\,e^{2r} = (\bar n - \sinh^2 r)\,e^{2r}$. Setting $\partial_r\meff = 0$ gives
\begin{equation}
\begin{aligned}
 & r^*(\bar n) = \tfrac{1}{2}\ln(1 + 2\bar n),
  \\
  & N^*(\bar n) = \frac{\bar n(1+\bar n)}{1+2\bar n},
  \\
  & \meff^*(\bar n) = \bar n(1+\bar n).
  \label{eq:r-star}
\end{aligned}
\end{equation}

$r^*(\bar n)$ is the squeezing that maximises the distinguishability of the two source states in the usual binary state-discrimination sense~\cite{Helstrom1976}. For QRNG randomness extraction the operational direction is the opposite: large $\meff$ means small overlap $\delta = e^{-2\meff}$ and small honest error rate $q_{\rm hon}$, both of which push the operating point toward the strong-feasibility boundary $q_b$ where $\HSh$ is small. Numerical analysis along the honest BPSK locus (see below) gives $\HSh$ monotonically decreasing in $\meff$ over the rate-positive interval, so $r^*$ is in fact the rate-minimising squeezing at fixed $\bar n$. We retain $r^*$ as a reference point because it locates the bottom of the $\HSh$-vs-$r$ curve, but the QRNG-relevant squeezing under the practical constraint $N \ge 0.05$ lies at one of the endpoints.

The lossless operating points and the corresponding lossy rates at $\eta=0.7$ are collected in Table~\ref{tab:summary-rates}.

Figure~\ref{fig:ushape} plots $\HSh$ against $r$ at fixed $\bar n$; the curves are U-shaped with the minimum at $r^*$ and the rate increasing toward both endpoints $r = 0$ and $r \to r_{\max}(\bar n) = \sinh^{-1}\sqrt{\bar n}$, the latter corresponding to $N \to 0$.

\begin{figure}
  \centering
   \includegraphics[width=\columnwidth]{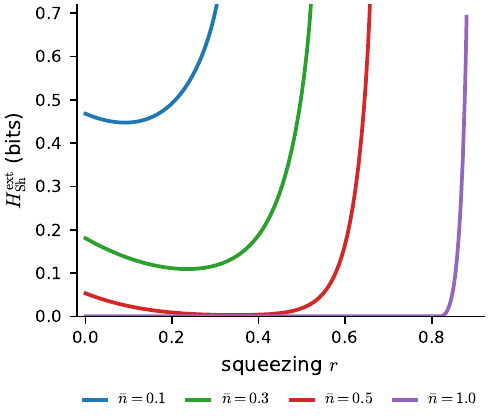}
  \caption{$\HSh$ along the honest BPSK locus vs squeezing $r$, at fixed $\bar n \in \{0.1, 0.3, 0.5, 1.0\}$; lossless source. Each curve runs from $r = 0$ to $r_{\max}(\bar n) = \sinh^{-1}\sqrt{\bar n}$ where $N \to 0$. The regime $\meff < 0.019$ in which $\HSh$ is an upper bound on $\Hcert$ rather than the certified rate occurs only on the extreme high-$r$ tails, below the BPSK-meaningful cutoff $N = 0.05$, and is not used for any certified operating point.}
  \label{fig:ushape}
\end{figure}

Numerical differentiation of $\HSh$ over a $10^3$-point grid in $\meff \in (0, \meff^{(z)})$ with $\meff^{(z)} \approx 0.844$ along the honest BPSK locus gives a strictly negative derivative throughout the grid, including points adjacent to both ends of the positive-rate branch, with limit $\HSh \to ((\pi + 2)/(2\pi))\,h(2/(\pi + 2)) \approx 0.789$ bits as $\meff \to 0^+$. The monotonicity is stated as a numerical observation only; no analytic proof of $d\HSh/d\meff < 0$ is given here.

For $\meff < 0.019$ the closed form is an upper bound on $\Hcert$ and not the certified rate, by Theorem~\ref{thm:main}(E) and Appendix~\ref{app:certified-region}. The relevant regime is at very small $N$ (i.e.\ $\bar n$ dominated by squeezing, $N e^{2r} < 0.019$) and is not reached by any operating point in Table~\ref{tab:summary-rates}.

\section{Loss tolerance and squeezing advantage}\label{sec:loss}

The Gram overlap $\delta = e^{-2 N e^{2r}}$ is a property of the source only and does not depend on $\eta$. The honest error probability degrades with loss according to Eq.~\eqref{eq:q-honest-lossy}, with $\meff^{\rm loss} = \eta N / (\eta e^{-2r} + 1 - \eta) < \meff$ for $\eta < 1$, so the lossy operating point sits at the same $\delta$ as the lossless one but at a larger $q$, hence closer to the rate-zero plateau.

Figure~\ref{fig:phase} plots ${\HSh}^{\max}(\bar n, \eta) := \max_r \HSh$ over the squeezing parameter and the additive squeezing benefit $\Delta H := {\HSh}^{\max} - \HSh|_{r = 0}$, on the $(\bar n, \eta)$ plane, with the optimisation restricted to $N \ge 0.05$ and to points satisfying the dual-feasibility test $m_2(\delta,q;5\times10^{-3})\ge0$.

\begin{table}[t]
  \centering
  \scriptsize
  \setlength{\tabcolsep}{2.2pt}
  \resizebox{\linewidth}{!}{%
  \begin{tabular}{rccccccc}
    \toprule
    $\bar n$
    & $r^{*}$
    & ${\meff}^{*}$
    & ${\delta}^{*}$
    & ${\HSh}^{*}$
    & $\left(\HSh\right)^{\max}_{\eta=1}$
    & $\left.\HSh\right|_{\eta=0.7,r=0}$
    & $\left(\HSh\right)^{\max}_{\eta=0.7}$ \\
    \midrule
    0.1 & 0.091 & 0.110 & 0.802 & 0.448 & 0.516 & 0.300 & 0.300\\
    0.3 & 0.235 & 0.390 & 0.458 & 0.109 & 0.407 & 0.063 & 0.115\\
    0.5 & 0.347 & 0.750 & 0.223 & 0.003 & 0.332 & 0.000 & 0.034\\
    1.0 & 0.549 & 2.000 & 0.018 & 0.000 & 0.199 & 0.000 & 0.000\\
    \bottomrule
  \end{tabular}%
  }
  \caption{Summary of the lossless and lossy operating points, in bits per symbol. The columns $r^*$, ${\meff}^{*}$, ${\delta}^{*}$ and ${\HSh}^{*}$ refer to the distinguishability-maximising squeezing at lossless honest BPSK; the $\bar n=1.0$ row lies on the rate-zero plateau at this point. The column $\left(\HSh\right)^{\max}_{\eta=1}$ is the lossless rate maximised over $r$ under the practical constraint $N\ge0.05$. The last two columns give the lossy no-squeezing baseline and the lossy maximum $\left(\HSh\right)^{\max}_{\eta=0.7}$, with the same $N\ge0.05$ constraint and the dual-feasibility condition $m_2(\delta,q;5\times10^{-3})\ge0$.}
  \label{tab:summary-rates}
\end{table}

\subsection{Source-trust sensitivity}

The trust model of Section~\ref{sec:setup} assumes that the user knows the Gram overlap $\delta$ of the two source states, and this is the minimal source information used by the rate formula: any pair of pure states with the stated overlap is unitarily equivalent to the canonical qubit pair on the source two-dimensional subspace, while the squeezed-coherent BPSK realisation of Section~\ref{sec:setup} is only one way of fixing $\delta$ through Eq.~\eqref{eq:overlap}. The trust assumption is therefore better described as ``pure-state pair with trusted Gram overlap'', rather than loosely as ``full source trust''.

A more conservative source assumption is photon-number trust, where the user trusts the photon-number statistics of the prepared state but not its off-diagonal coherence~\cite{Lasota2023}. As an empirical sensitivity check, we solved the corresponding photon-number-and-phase SDP for $\bar n \in [0.05,0.7]$ and $\eta \in [0.4,1.0]$ and found that the squeezing benefit $\Delta H$ of Section~\ref{sec:loss} disappears, with the certified rate depending on $\bar n$ alone. This scan is reported only as an empirical robustness check, not as a closed-form theorem.

\section{Discussion}\label{sec:discussion}
\begin{figure*}
  \centering
  \includegraphics[width=\linewidth]{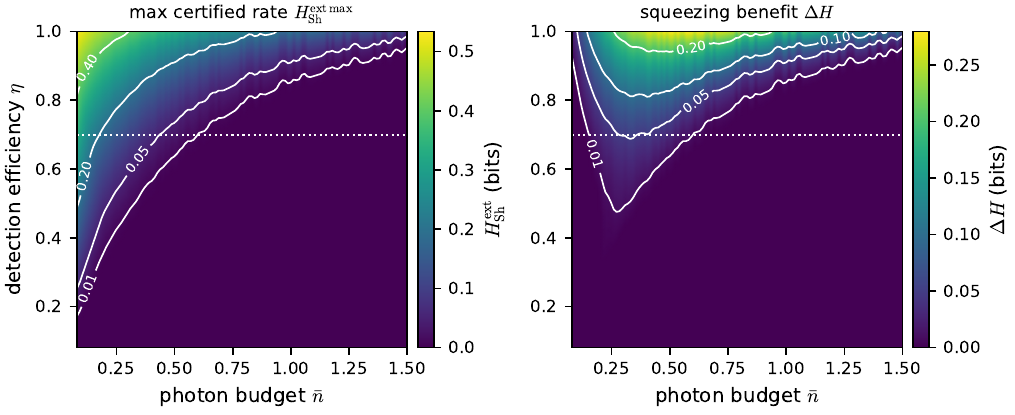}
  \caption{Left: maximum certified rate ${\HSh}^{\max}(\bar n, \eta) := \max_r \HSh$ subject to $N \ge 0.05$ and the dual-feasibility condition $m_2(\delta,q;5\times10^{-3})\ge0$. Right: additive squeezing benefit $\Delta H := {\HSh}^{\max} - \HSh|_{r = 0}$. Honest BPSK observation under Eq.~\eqref{eq:q-honest-lossy}. Dotted line: $\eta = 0.7$.}
  \label{fig:phase}
\end{figure*}
This work gives a closed-form Shannon-rate formula for semi-DI randomness generation with a trusted binary pure-state source and an untrusted binary detector carrying only classical side information. The main point is not just the squeezed-state application, but the fact that the full binary-qubit POVM optimisation can be solved once the deterministic endpoints $M_0=0$ and $M_0=I$ are included. These endpoints are missed by the projective-only treatment, but they change the adversarial decomposition and can reduce the certified rate substantially, as seen already at $\delta=0.5$ and $q=0.15$. The resulting expression $\HSh(\delta,q)$ is therefore the relevant extended-LP rate on the verified operating region, and it gives a direct analytic handle on how the trusted source overlap and the observed error probability determine the extractable Shannon randomness.

For squeezed-coherent BPSK, the formula exposes a simple but useful design rule. Squeezing aligned with the displacement axis increases the distinguishability through $\meff=Ne^{2r}$, but in this semi-DI randomness setting higher distinguishability is not automatically beneficial, because it also drives the honest operating point toward the low-randomness boundary of the feasible region. Thus the distinguishability-maximising squeezing $r^*(\bar n)=\tfrac12\ln(1+2\bar n)$ is not the randomness-maximising choice at fixed energy. This gives the squeezed-state source a clear operational role beyond a communication-style discrimination advantage, since the useful regime is set by the balance between overlap, observed error, loss, and the dual-feasibility condition. The construction is analytic on the primal side and dual-certified on the numerically verified region $\Rcert$, which contains all operating points used here; outside this region the same expression remains a valid upper bound but is not claimed as tight. The assumptions are correspondingly explicit, namely a pure-state source with trusted Gram overlap, verified symmetry of the two BPSK error marginals, and classical detector side information. Within that model the result gives a compact rate formula for squeezed-state semi-DI QRNGs and a basis for the next step, where one can replace the numerical dual check by an analytic certificate, add stronger source constraints, or extend the LP geometry from BPSK to QPSK and general $m$PSK.
\newpage
\appendix

\section{Threat model and adversary scope}\label{app:threat}

\subsection{Classical-$\lambda$ Shannon objective}

Let $\Lambda$ be the classical random variable representing the value of the hidden parameter $\lambda$ that the adversary holds, with distribution $p_\lambda$. Conditional on $X = 0$ and $\Lambda = \lambda$, the device produces $B \in \{0, 1\}$ with $p_{B|0, \lambda}(b) = \langle\psi_0 | M_{b|\lambda} | \psi_0\rangle$, and the conditional entropy of $B$ given $X = 0, \Lambda = \lambda$ is $h(p_{B|0, \lambda}(0))$. Averaging over $\lambda$,
\begin{equation}
  H(B \mid X = 0, \Lambda) = \sum_\lambda p_\lambda\,h\!\bigl(\langle\psi_0|M_{0|\lambda}|\psi_0\rangle\bigr),
\end{equation}
and an analogous identity for $X = 1$. The certified asymptotic i.i.d.\ Shannon rate against the classical-$\lambda$ adversary considered in this paper is the minimum of the $X=0$ quantity over $\{p_\lambda, M_{b|\lambda}\}$ subject to Eq.~\eqref{eq:marginals}. By the symmetry $X \leftrightarrow 1-X$ the corresponding one-input optimisation for $X=1$ has the same value, but the averaged quantity $H(B|X,\Lambda)$ is a different optimisation and is not solved here. No environment register is introduced: $\Lambda$ is classical, and the side information is the label $\lambda$ only.

\subsection{Detector-purification adversary}

A strictly stronger adversary holds the purifying environment of the binary detector rather than only the classical label $\lambda$. The detector is then described by an isometry $V: \mathbb{C}^2 \to \mathbb{C}^2 \otimes \mathcal{H}_E$ with $V^\dagger V = I$ and
\begin{equation}
  V\ket{\psi_x} = \sum_b K_b\ket{\psi_x}\otimes\ket{e_b}_E,
  \quad K_b = \sqrt{M_b},
\end{equation}
where $\{\ket{e_b}_E\}$ is an orthonormal basis of an outcome-marking environment factor. The bipartite outcome--environment state
\begin{equation}
  \sigma_{BE\mid X = x} = \sum_b p_{B|x}(b)\,\ket b\bra b_B\otimes\ket{e_b}\bra{e_b}_E
\end{equation}
has $H(B \mid E, X = x) = 0$ for any input $x$ and any binary POVM $\{M_b\}$, since outcome and environment are perfectly correlated. Eve holding $E$ recovers $B$ perfectly and the certified rate is zero on all observations.

The two cases are distinct adversary models. The detector-purification adversary has access to an orthogonal outcome-tagging quantum register that does not exist in the classical-$\lambda$ model; including such a register in the classical-$\lambda$ analysis would conflate them. The closed form of Theorem~\ref{thm:main} is a randomness guarantee against the classical-$\lambda$ adversary only. The trust assumption that distinguishes the two cases is that side information is classical.

\section{Strong feasibility and the honest BPSK locus}\label{app:strong-feas}

The LP marginals $(u, v) := (\langle\psi_0|M_0|\psi_0\rangle, \langle\psi_1|M_0|\psi_1\rangle)$ for a projective at angle $\theta$ are $u = \cos^2\theta$, $v = \cos^2(\theta - \theta_*)$ with $\theta_* = \arccos\delta$. The locus of feasible $(u, v)$ as $\theta$ runs over $[0, 2\pi)$ is the moment ellipse
\begin{equation}
  u^2 + v^2 + 2(1 - 2\delta^2)\,uv - 2(1 - \delta^2)(u + v) + (1 - \delta^2)^2 = 0,
  \label{eq:ellipse}
\end{equation}
equivalently in centred coordinates $u' := 2u - 1$, $v' := 2v - 1$,
\begin{equation}
  u'^2 + v'^2 - 2(2\delta^2 - 1)\,u'v' = 4\delta^2(1 - \delta^2).
  \label{eq:ellipse-centred}
\end{equation}
Equation~\eqref{eq:ellipse-centred} is obtained by substituting $u' = \cos(2\theta)$, $v' = \cos(2\theta - 2\theta_*)$, using $v' = \cos(2\theta_*)\,u' + \sin(2\theta_*)\sin(2\theta)$ and eliminating $\sin(2\theta)$ via $\sin^2(2\theta) = 1 - u'^2$, with $\cos(2\theta_*) = 2\delta^2 - 1$ and $\sin(2\theta_*) = 2\delta\sqrt{1 - \delta^2}$.

The deterministic endpoints $(0, 0)$ and $(1, 1)$ lie outside the ellipse on the line $v = u$. The feasibility region for the binary qubit POVM is the convex hull of the ellipse and the two components. Substituting $v = 1 - u$ (the BPSK chord) into Eq.~\eqref{eq:ellipse} gives the quadratic $u^2 - u + \delta^2/4 = 0$ with solutions $u = (1 \pm \sqrt{1 - \delta^2})/2$. The honest observation $(1 - q, q)$ is feasible iff
\begin{equation}
  |1 - 2q| \le \sqrt{1 - \delta^2},
  \qquad\text{i.e.,}\ q_b \le q \le 1 - q_b,
\end{equation}
with $q_b$ as in Eq.~\eqref{eq:qb}.

To verify that the honest BPSK locus $\delta = e^{-2\meff}$, $q = (1 - \mathrm{erf}\sqrt{2\meff})/2$ lies strictly inside the strong-feasibility interval, set $x := \sqrt{2\meff}$ so that $\delta = e^{-x^2}$ and $1 - 2q = \mathrm{erf}(x)$. The condition becomes
\begin{equation}
  \mathrm{erf}^2(x) \le 1 - e^{-2x^2}.
  \label{eq:erf-vs-exp}
\end{equation}
Define $s(x) := 1 - e^{-2x^2} - \mathrm{erf}^2(x)$. Then $s(0) = 0$ and $s(x) \to 0$ as $x \to \infty$. Differentiating,
\begin{equation}
  s'(x) = 4 e^{-x^2}\,g(x), \qquad g(x) := x e^{-x^2} - \frac{\mathrm{erf}(x)}{\sqrt\pi}.
\end{equation}
We have $g(0) = 0$ and
\begin{equation}
  g'(x) = e^{-x^2}\Bigl(1 - \tfrac{2}{\pi} - 2x^2\Bigr),
\end{equation}
which vanishes at $x_c := \sqrt{(\pi - 2)/(2\pi)} \approx 0.4262$, with $g' > 0$ on $(0, x_c)$ and $g' < 0$ on $(x_c, \infty)$. Hence $g$ rises from $g(0) = 0$ to its maximum at $x_c$ and decreases thereafter, with $g(x) \to -1/\sqrt\pi$ as $x \to \infty$. By the intermediate-value theorem $g$ has a single positive zero $x_+$ on $(x_c, \infty)$, computed numerically as $x_+ \approx 0.806$. On $(0, x_+)$, $g > 0$ hence $s' > 0$; on $(x_+, \infty)$, $g < 0$ hence $s' < 0$. Combined with $s(0) = s(\infty) = 0$, this gives $s > 0$ on $(0, \infty)$, which is Eq.~\eqref{eq:erf-vs-exp} strictly. The critical point $x_c$ of $g$ and the positive zero $x_+$ are distinct: $g$ reaches its maximum at $x_c$ and crosses zero at $x_+$.

\section{Analytic proof of $\alpha + \beta \le 0$}\label{app:slack-sign}

We work in the convention of Lemma~\ref{lem:dual}: active deterministic endpoint $(0, 0)$ (so $\gamma = 0$), slack at $(1, 1)$ to be established as $\alpha + \beta \le 0$, with $p = \cos^2\theta_0 \in (1 - q_b, 1) \subset (1/2, 1)$ on the rate-positive interior.

From Cramer's rule applied to Eqs.~\eqref{eq:dual-val}, \eqref{eq:dual-slope} and the determinant $D$ of Eq.~\eqref{eq:det-formula},
\begin{equation}
  \alpha + \beta = \frac{F(\theta_0, \theta_*)}{D},
\end{equation}
with
\begin{align}
  F &= h(p)\,\bigl[\sin(2(\theta_0 - \theta_*)) - \sin(2\theta_0)\bigr] \notag\\
  &\quad+ h'(p)\sin(2\theta_0)\bigl[\cos^2\theta_0 - \cos^2(\theta_0 - \theta_*)\bigr].
\end{align}
Using $\sin(2(\theta_0 - \theta_*)) - \sin(2\theta_0) = -2\sin\theta_*\cos(2\theta_0 - \theta_*)$ and $\cos^2\theta_0 - \cos^2(\theta_0 - \theta_*) = -\sin\theta_*\,\sin(2\theta_0 - \theta_*)$, then $\sin(2\theta_0) = -2\sqrt{p(1-p)}$ (the branch sign of Eq.~\eqref{eq:theta-0}),
\begin{equation}
\begin{split}
  F
  &= -2\sin\theta_*\,h(p)\cos(2\theta_0-\theta_*) \\
  &\quad
  +2\sqrt{p(1-p)}\,\sin\theta_*\,h'(p)
  \sin(2\theta_0-\theta_*) \\
  &=
  -2\sin\theta_*
  \bigl[
    h(p)\cos(2\theta_0-\theta_*) \\
  &\qquad\qquad
    -\sqrt{p(1-p)}\,h'(p)\sin(2\theta_0-\theta_*)
  \bigr].
\end{split}
\end{equation}

Expanding $\cos(2\theta_0 - \theta_*) = (2p - 1)\delta - 2\sqrt{p(1-p)}\sqrt{1 - \delta^2}$ and $\sin(2\theta_0 - \theta_*) = -2\sqrt{p(1-p)}\,\delta - (2p - 1)\sqrt{1 - \delta^2}$ for $\theta_0 = \pi - \arccos\sqrt p$, and collecting,
\begin{equation}
  F = -2\sin\theta_*\bigl[\delta\,G(p) - 2\sqrt{(1-\delta^2)\,p(1-p)}\,\mathcal H(p)\bigr],
\end{equation}
with
\begin{align}
  G(p) &:= (2p - 1)\,h(p) + 2p(1-p)\,h'(p),\\
  \mathcal H(p) &:= h(p) - \tfrac{1}{2}(2p - 1)\,h'(p).
\end{align}

In nats, $h_{\rm nat}(p) = -p\ln p - (1-p)\ln(1-p)$, $h_{\rm nat}'(p) = \ln((1-p)/p)$, and direct computation gives
\begin{align}
  G(p)\ln 2 &= -p\ln p + (1-p)\ln(1-p),\\
  \mathcal H(p)\ln 2 &= -\tfrac{1}{2}\ln\bigl(p(1-p)\bigr).
\end{align}
The function $\widetilde G(p) := G(p)\ln 2$ is antisymmetric about $p = 1/2$: $\widetilde G(1 - p) = -\widetilde G(p)$. It vanishes at $p = 0$, $p = 1/2$, $p = 1$, and its second derivative is $\widetilde G''(p) = (2p - 1)/(p(1 - p))$, positive on $(1/2, 1)$ and negative on $(0, 1/2)$. Hence $\widetilde G$ is strictly convex on $(1/2, 1)$, vanishes at both endpoints of that subinterval, and is strictly negative there. So $G(p) < 0$ on the rate-positive interior $p \in (1 - q_b, 1) \subset (1/2, 1)$.

The function $\mathcal H(p)$ is symmetric: $\mathcal H(1 - p) = \mathcal H(p)$. It is strictly positive on $(0, 1)$ because $p(1-p) \le 1/4 < 1$ gives $-\ln(p(1-p)) > 0$. So $\mathcal H(p) > 0$ on the rate-positive interior.

Putting the signs together: $\sin\theta_* > 0$, $\delta > 0$, $\sqrt{(1-\delta^2)p(1-p)} > 0$, $G(p) < 0$, $\mathcal H(p) > 0$, so $\delta G(p) - 2\sqrt{(1-\delta^2)p(1-p)}\,\mathcal H(p) < 0$, hence $F > 0$. The determinant $D = -2\sin\theta_*\,\cos\theta_0\,\cos(\theta_0 - \theta_*)$ in Eq.~\eqref{eq:det-formula} has $\sin\theta_* > 0$, $\cos\theta_0 = -\sqrt p < 0$, and $\cos(\theta_0 - \theta_*) = -\delta\sqrt p + \sqrt{(1-p)(1-\delta^2)}$; for $p \in (1 - q_b, 1)$ the first term dominates and $\cos(\theta_0 - \theta_*) < 0$. Hence $D < 0$ strictly. Combining, $\alpha + \beta = F/D < 0$ strictly on the rate-positive interior.

At the rate-zero edge $q \to q_z^-$, $p \to 1$, $G(1) = 0$ and $\sqrt{p(1-p)}\,\mathcal H(p) \to 0$ (the $\sqrt{1-p}$ dominates the logarithmic divergence of $\mathcal H$), so $F \to 0$. Here $r \to (1+\delta^2)^{-1}$, not zero, so the dual value $r\cdot h(p)$ vanishes only because $h(p) \to h(1) = 0$, while $\alpha + \beta \to 0^-$ and the deterministic slack becomes tight at the plateau boundary.

The mirror $q \in (1 - q_z, 1 - q_b)$ uses active component $(1, 1)$ and the slack condition $\gamma \le 0$ instead; the proof is identical under $p \to 1 - p$ on the antisymmetric $G$ and symmetric $\mathcal H$.

\section{The dual-feasibility region}\label{app:certified-region}

The dual-feasibility region of Theorem~\ref{thm:main}(E) is the set of $(\delta, q)$ such that the residual
\begin{equation}
  g(\theta; \delta, q) := h(\cos^2\theta) - L_{\delta, q}(\theta)
\end{equation}
is non-negative for every $\theta \in [0, \pi)$. By construction $g(\theta_0) = 0$ and $g'(\theta_0) = 0$ (value and slope tangency at the active projective angle), so the tangent point $\theta_0$ contributes a double zero to the residual; this zero is present everywhere in $\Rcert$ and does not characterise the boundary. Inside the interior of $\Rcert$, $g$ is strictly positive at every $\theta \neq \theta_0$; on $\partial\Rcert$, $g$ acquires an additional tangent zero at some $\theta_1 \neq \theta_0$ (where the second extremum of $g$ first touches zero from above); beyond $\partial\Rcert$, $g$ is strictly negative on a neighbourhood of $\theta_1$.

For numerical detection, fix a small exclusion radius $\varepsilon > 0$ around $\theta_0$ and compute
\begin{equation}
  m_2(\delta, q;\varepsilon) := \min_{\theta\,:\,|\theta - \theta_0| > \varepsilon}\,g(\theta; \delta, q).
\end{equation}
For $\varepsilon$ smaller than half the distance from $\theta_0$ to the nearest critical point of $g$ outside the tangent neighbourhood, $m_2 > 0$ strictly in the interior of $\Rcert$, $m_2 = 0$ on $\partial\Rcert$, and $m_2 < 0$ in $\Rcert^c$; the boundary location is insensitive to $\varepsilon$ in this range. The numerical scan reported below uses $\varepsilon = 5 \times 10^{-3}$, which is well below the observed separation between $\theta_0$ and the secondary critical point of $g$ for every $(\delta, q)$ tested.

Substituting $u = \cos(2\theta)$, $v = \sin(2\theta)$ into the candidate witness gives $L_{\delta, q}(\theta) = c_0 + c_1\,u + c_2\,v$ with
\begin{equation}
  c_0 = \tfrac{\alpha + \beta}{2} + \gamma,
  \quad
  c_1 = \tfrac{\alpha}{2} + \tfrac{\beta}{2}\cos(2\theta_*),
  \quad
  c_2 = \tfrac{\beta}{2}\sin(2\theta_*).
  \label{eq:c012}
\end{equation}
Define the tangent line to $\Phi(u) := h((1+u)/2)$ at $u_0 = 2p - 1$ by $\ell_*(u) := \Phi(u_0) + \Phi'(u_0)(u - u_0)$. Combining the tangency conditions at $\theta_0$ with $\gamma = 0$, one obtains
\begin{equation}
  L_{\delta, q}(\theta) = \ell_*\bigl(u(\theta)\bigr) - \frac{2 c_2}{v_0}\,\sin^2(\theta - \theta_0),
  \label{eq:L-as-tangent}
\end{equation}
with $v_0 := \sin(2\theta_0)$. In the active branch of Eq.~\eqref{eq:theta-0}, $v_0 = -2\sqrt{p(1-p)} < 0$; the dual coefficient $c_2 = (\beta/2)\sin(2\theta_*)$ has $\beta < 0$ at every rate-positive point (verified directly from the linear system of Lemma~\ref{lem:dual}), so $c_2 < 0$ as well, and the ratio $2c_2/v_0 > 0$.

Strict concavity of $\Phi$ identifies the tangent--curve gap with the binary Kullback--Leibler divergence in bits~\cite{CoverThomas2006}:
\begin{equation}
  \ell_*(u) - \Phi(u) = \dKL\bigl((1+u)/2\,\big\|\,p\bigr).
\end{equation}
The pointwise inequality $L_{\delta, q}(\theta) \le h(\cos^2\theta)$ is therefore equivalent to
\begin{equation}
  \dKL\bigl(\cos^2\theta\,\big\|\,p\bigr) \le \frac{2 c_2}{v_0}\,\sin^2(\theta - \theta_0)\quad \forall\,\theta \in [0, \pi),
  \label{eq:K-le-R}
\end{equation}
with both sides non-negative on the rate-positive interior (the ratio $2c_2/v_0$ being positive by the sign analysis above). Both sides vanish to second order at $\theta = \theta_0$, with local Taylor ratio
\begin{equation}
  \lim_{\theta \to \theta_0}\frac{\dKL\bigl(\cos^2\theta\,\big\|\,p\bigr)}{\sin^2(\theta - \theta_0)} = \frac{2}{\ln 2}.
\end{equation}
\begin{figure}
  \centering
  \includegraphics[width=\columnwidth]{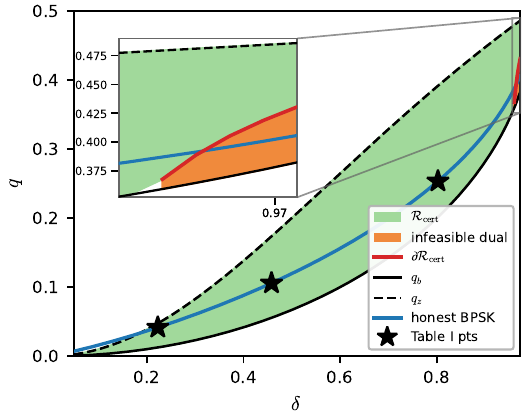}
    \caption{Dual-feasibility region $\Rcert$ (green) and its complement in the rate-positive interior (orange). The boundary $\partial\Rcert$ (red) is the numerically computed curve $m_2(\delta, q) = 0$ with the tangent-contact neighbourhood excluded. Curves: $q_b$ strong-feasibility lower edge (solid black), $q_z$ rate-zero upper edge (dashed black), honest BPSK $q = (1 - \mathrm{erf}\sqrt{2\meff})/2$ at $\delta = e^{-2\meff}$ (blue). Black stars: lossless operating points from Table~\ref{tab:summary-rates} at $\bar n \in \{0.1, 0.3, 0.5\}$. Inset: zoom near the high-$\delta$ boundary, showing where the honest BPSK locus exits $\Rcert$ close to $\delta_c\simeq0.963$.}
  \label{fig:Rcert}
\end{figure}

The local Taylor expansion gives the necessary condition
\begin{equation}
\frac{2 c_2}{v_0} \ge \frac{2}{\ln 2}, \quad\text{i.e.,}\quad c_2\ln 2 \le v_0,
  \label{eq:local-cond}
\end{equation}
where the inequality is between signed quantities and equality marks the boundary of local feasibility. In the active branch with $v_0 < 0$ and $c_2 < 0$, the condition is equivalent to $|c_2|/|v_0| \ge 1/\ln 2$. Equation~\eqref{eq:local-cond} is necessary but not sufficient for Eq.~\eqref{eq:K-le-R}; non-local features of $g$ outside the tangent neighbourhood can violate the global inequality even when the local condition holds.

The global inequality has been verified numerically on a $(\delta, q)$ grid of spacing $\Delta\delta = \Delta q = 0.01$ across $\delta \in [0.05, 0.95]$, $q \in (q_b, q_z)$, evaluating $g$ on a $2 \times 10^4$-point grid in $\theta$ and excluding a small neighbourhood of the tangent contact $\theta_0$; the worst $m_2$ is non-negative at the $10^{-7}$ level. A rigorous Lipschitz-grid certificate based on an explicit bound on $|g''(\theta)|$ would convert this into an analytic lower bound on $m_2$ in the spirit of \cite{Boyd2004}; we do not carry out the bound here. Accordingly we describe $\Rcert$ as \emph{numerically verified} at the precision stated, not analytically certified, and leave the analytic certificate for future work.

For honest BPSK with $\delta = e^{-2\meff}$ and $q = (1 - \mathrm{erf}\sqrt{2\meff})/2$, the global inequality~\eqref{eq:K-le-R} first fails at $\delta_c \approx 0.963$ (i.e.\ $\meff \approx 0.019$); the necessary local condition~\eqref{eq:local-cond} alone continues to hold out to $\delta \approx 0.98$, so it is the non-local violation away from $\theta_0$ that sets the boundary. The numerically observed boundary of $\Rcert$ along the honest locus coincides with $\delta_c$ to the grid precision. Below $\delta_c$ the closed form is the certified rate; above $\delta_c$ it is strictly larger than $\Hcert$. The regime $\delta > \delta_c$ corresponds to $\meff < 0.02$, well below any operating point in Table~\ref{tab:summary-rates}.
Figure~\ref{fig:Rcert} shows the resulting dual-feasibility region $\Rcert$ together with its numerically obtained boundary in the $(\delta,q)$ plane.

Inside $\Rcert$ the optimal adversary uses one projective component and the deterministic endpoint $(0, 0)$, the structure of Lemma~\ref{lem:primal}: the dual is feasible at every $\theta$, weak duality combined with primal achievability forces equality, and the closed form equals $\Hcert$. Outside $\Rcert$, the optimal adversary uses two projective components with no deterministic endpoint; this is a structurally different LP solution that the closed form is not built on, and the closed form is then a strict upper bound. The transition between the two regimes is sharp, occurring at $m_2 = 0$.

\bibliography{Ref}

\end{document}